%% file: main.tex
\documentclass[fullpage,11pt]{article}

\usepackage{amsthm}
\usepackage{graphicx} 
\usepackage{array} 

\usepackage{amsmath, amssymb, amsfonts, verbatim}
\usepackage{hyphenat,epsfig,subfigure,multirow}

\usepackage[usenames,dvipsnames]{xcolor}
\usepackage[ruled]{algorithm2e}

\usepackage{tcolorbox}

\definecolor{DarkRed}{rgb}{0.5,0.1,0.1}
\definecolor{DarkBlue}{rgb}{0.1,0.1,0.5}

\usepackage{hyperref}
\hypersetup{
colorlinks=true,
pdfnewwindow=true,
citecolor=ForestGreen,
linkcolor=DarkRed,
filecolor=DarkRed,
urlcolor=DarkBlue
}

\usepackage{bm}
\usepackage{url}
\usepackage{xspace} 
\usepackage[mathscr]{euscript}

\usepackage{mdframed}

\usepackage[noend]{algpseudocode}
\makeatletter
\def\BState{\State\hskip-\ALG@thistlm}
\makeatother

\usepackage{cite}
\usepackage{enumerate}

\usepackage[margin=1in]{geometry}

\newtheorem{theorem}{Theorem}
\newtheorem{lemma}{Lemma}[section]

\newtheorem{claim}[lemma]{Claim}

\newtheorem{definition}{Definition}

\newtheorem{problem}{Problem}
\newtheorem{remark}[lemma]{Remark}

\newtheorem*{claim*}{Claim}
\newtheorem*{proposition*}{Proposition}
\newtheorem*{lemma*}{Lemma}
\newtheorem*{problem*}{Problem}

\newtheorem{mdresult}{Result}

\allowdisplaybreaks

\renewcommand{\qed}{\nobreak \ifvmode \relax \else
      \ifdim\lastskip<1.5em \hskip-\lastskip
      \hskip1.5em plus0em minus0.5em \fi \nobreak
      \vrule height0.75em width0.5em depth0.25em\fi}

\usepackage[T1]{fontenc}
\usepackage[utf8]{inputenc}

\input{macros}

\title{Simple Round Compression for Parallel Vertex Cover}
\author{Sepehr Assadi\\University of Pennsylvania \\ \texttt{sassadi@cis.upenn.edu}}

\date{}

\begin{document}
\maketitle

\input{abstract}
\setcounter{page}{1}

\input{intro}

\input{prelim}

\input{vc-PA}

\subsection*{Acknowledgements}
I am grateful to my advisor Sanjeev Khanna for our collaboration in~\cite{AssadiK17} that formed a crucial component of the current work and to Sanjeev Khanna and Krzysztof Onak  
for feedback on a draft of this note. I also thank Omri Weinstein and Qin Zhang for helpful discussions.

\bibliographystyle{abbrv}
\bibliography{general}


\end{document}

%% file: macros.tex
\newcommand{\toShrink}{-.20cm}
\newcommand{\toShrinkEnu}{-.2cm}

\newcommand{\Ot}{\ensuremath{\widetilde{O}}}
\newcommand{\eps}{\ensuremath{\varepsilon}}

\newcommand{\Bracket}[1]{\Big[#1\Big]}

\newcommand{\paren}[1]{\ensuremath{\left(#1\right)}\xspace}
\newcommand{\card}[1]{\left\vert{#1}\right\vert}

\newcommand{\ceil}[1]{{\left\lceil{#1}\right\rceil}}

\newcommand{\set}[1]{\ensuremath{\left\{ #1 \right\}}}
\newcommand{\poly}{\mbox{\rm poly}}

\DeclareMathOperator*{\Exp}{\ensuremath{{\mathbb{E}}}}

\newcommand{\Ex}{\Exp}

\newcommand{\etal}{{\it et al.\,}}

\newenvironment{tbox}{\begin{tcolorbox}[
		enlarge top by=5pt,
		enlarge bottom by=5pt,
		 boxsep=0pt,
                  left=4pt,
                  right=4pt,
                  top=10pt,
                  arc=0pt,
                  boxrule=1pt,toprule=1pt,
                  colback=white
                  ]
	}
{\end{tcolorbox}}

\newcommand{\textbox}[2]{
{
\begin{tbox}
\textbf{#1}
{#2}
\end{tbox}
}
}

\newcommand{\bO}{\overline{O}}

\newcommand{\Gi}[1]{\ensuremath{G^{(#1)}}}

\newcommand{\Vi}[1]{\ensuremath{V^{(#1)}}}

\newcommand{\Vvc}{\ensuremath{O^{\star}}}
\newcommand{\bVvc}{\ensuremath{\overline{\Vvc}}}

\newcommand{\vc}{\ensuremath{\textnormal{\textsf{opt}}}}



%% file: abstract.tex
\begin{abstract} 

Recently, Czumaj~\etal~(arXiv 2017) presented a parallel (almost) $2$-approximation algorithm for the maximum
matching problem in only $O\paren{(\log\log{n})^2}$ rounds of the \emph{massive parallel computation (MPC)} framework, when the memory per machine is $O(n)$. The main 
approach in their work is a way of \emph{compressing} $O(\log{n})$ rounds of a distributed algorithm for maximum matching into only $O\paren{(\log\log{n})^2}$ MPC rounds. 

In this note, we present a similar algorithm for the closely related problem of approximating the minimum vertex cover in the MPC framework. 
We show that one can achieve an $O(\log{n})$ approximation to minimum vertex cover in only $O(\log\log{n})$ MPC rounds 
when the memory per machine is $O(n)$. Our algorithm for vertex cover is similar to the maximum matching algorithm of Czumaj~\etal but avoids many of the intricacies in their approach and as a result admits
a considerably simpler analysis (at a cost of a worse approximation guarantee). 
We obtain this result by modifying a previous parallel algorithm by Khanna and the author (SPAA 2017) for vertex cover that allowed for compressing $O(\log{n})$ rounds of a 
distributed algorithm into \emph{constant} MPC rounds when the memory allowed per machine is $O(n\sqrt{n})$. 
 

\end{abstract}

%% file: intro.tex
\newcommand{\Alg}{\ensuremath{\textnormal{\textsf{ALG}}}\xspace}
\newcommand{\EE}{\ensuremath{\mathcal{E}}}
\renewcommand{\tilde}{\widetilde}

\section{Introduction} \label{INTRO}

Minimum vertex cover and closely related maximum matching problems are among the most well-studied classical optimization problems. 
Naturally, these problems have been studied in the \emph{massive parallel computation (MPC)} model of~\cite{KarloffSV10} 
that abstracts out the capability of many existing frameworks for parallel computation such as MapReduce and Hadoop. 

The first MPC algorithms for matching and vertex cover are due to Lattanzi~\etal~\cite{LattanziMSV11} and obtain $2$-approximation in $O(1)$ rounds and $n^{1+\Omega(1)}$ space per machine. 
The approximation guarantee for the matching problem was further improved to $(1+\eps)$ by Ahn and Guha~\cite{AhnG15}. Recently, it was shown by Khanna and the author that 
one can achieve an $O(1)$-approximation to maximum matching and $O(\log{n})$-approximation to minimum vertex cover in at most \emph{two} MPC rounds and $O(n\sqrt{n})$ space per machine~\cite{AssadiK17}. 
However, when the space allocated to each machine is $O(n)$, the performance of all these algorithms degrade to $\Omega(\log{n})$ rounds. 

In a recent breakthrough, Czumaj~\etal~\cite{CzumajLMMOS17} provided the first MPC algorithm for maximum matching that achieves an $O(1)$ approximation 
in only $O\paren{(\log\log{n})^2}$ rounds and $O(n)$ space per machine (even $n/(\log{n})^{O(\log\log{n})}$ space). While the high level idea of the algorithm in~\cite{CzumajLMMOS17} is  
natural in hindsight, the actual algorithm and analysis are quite intricate. In this note, we combine the ideas in~\cite{CzumajLMMOS17} and~\cite{AssadiK17} to provide a similar 
algorithm for the closely related problem of minimum vertex cover that admits a considerably simpler analysis. In particular, we prove that,

\begin{theorem}\label{thm:vc}
	There exists a randomized MPC algorithm that with high probability computes an $O(\log{n})$ approximation to minimum vertex cover in $O(\log\log{n})$ rounds, assuming that the memory per each machine is 
	$O(n)$.  
\end{theorem}

We remark that similar to~\cite{CzumajLMMOS17}, we can extend our result to the slightly sublinear space regime where the memory per machine is only $n/(\log{n})^{O(1)}$. It is also
worth mentioning that the algorithm of~\cite{CzumajLMMOS17} does \emph{not} work for the minimum vertex cover problem. In fact, extending
the results in~\cite{CzumajLMMOS17} to minimum vertex cover has been cast as an open question in~\cite{CzumajLMMOS17}. 


%% file: prelim.tex
\section{Preliminaries} \label{PRELIM}

\paragraph{Notation.} For any integer $t$, $[t] := \set{1,\ldots,t}$. Let $G(V,E)$ be a graph; $\vc(G)$ denotes the minimum vertex cover size in $G$. 
For vertices $S \subseteq V$ and a vertex $v \in V \setminus S$, $D(v,S)$ denotes the size of the intersection of neighbor-set of $v$ and the set $S$, i.e., the degree of $v$ to the set $S$. 
 
 \paragraph{MPC model.} We use the same model of parallel computation as in~\cite{CzumajLMMOS17}, namely the \emph{Massive Parallel Computation (MPC)} model, 
 that is a simple variant of the model introduced originally in~\cite{KarloffSV10} and was further refined in~\cite{GoodrichSZ11,BeameKS13,AndoniNOY14}. 
 
 In this model, there are $p$ machines each with a memory of size $s$ such that $p \cdot s = O(N)$; here, $N$ is the total memory required to represent the input. 
The computation proceeds in synchronous rounds: in each round, each machine performs some local computation and at the end of the round machines exchange messages to guide the computation for the next round. 
All messages sent and received by each machine in each round have to fit into the local memory of the machine, and hence their total length is bounded by $s$ in each round. 
At the end, machines collectively output the solution. The data output by each machine also has to fit in its local memory.

As we consider graph problems in this work, the input size can be as large as $\Theta(n^2)$ for a graph with $n$ vertices. 
Similar to~\cite{CzumajLMMOS17}, our main focus is on the case when the memory per each machine is linear, i.e., $s = O(n)$.

%% file: vc-PA.tex
\newcommand{\OPVC}{\ensuremath{\textnormal{\textsf{Original-Peeling-VC}}}\xspace}
\newcommand{\LP}{\ensuremath{\textnormal{\textsf{Local-Peeling}}}\xspace}
\newcommand{\PA}{\ensuremath{\textnormal{\textsf{Parallel-Peeling}}}\xspace}

\newcommand{\Pstar}{\ensuremath{P^{\star}}}

\renewcommand{\Pi}[1]{\ensuremath{P^{(#1)}}}
\newcommand{\Oi}[1]{\ensuremath{O^{(#1)}}}
\newcommand{\bOi}[1]{\ensuremath{\bO^{(#1)}}}
\newcommand{\Ai}[1]{\ensuremath{A^{(#1)}}}
\newcommand{\Bi}[1]{\ensuremath{B^{(#1)}}}

\newcommand{\tDelta}{\ensuremath{\widetilde{\Delta}}}

\newcommand{\Ostar}{\ensuremath{O^{\star}}}
\newcommand{\bOstar}{\ensuremath{\overline{O^{\star}}}}

\section{A Parallel Algorithm for Vertex Cover}\label{sec:vc-PA}

We present our parallel algorithm, \PA, for vertex cover in this section. We first provide a high level overview of our approach and compare our techniques with those of~\cite{AssadiK17,CzumajLMMOS17}. 
Next, we present the formal algorithm and then analyze its approximation ratio. We finish this section by providing the necessary details for implementing \PA in the MPC model. For clarity of the
exposition, we present our algorithm in its simplest form which requires $\Ot(n)$ memory per machine. At the end of this section, we provide the necessary details for extending this algorithm to the case when the space per machine is $O(n)$ or even slightly sublinear in $n$. 

\subsection{Overview of the Algorithm}\label{sec:vc-overview} 

The starting point of our algorithm is the following \emph{peeling process} of Parnas and Ron~\cite{ParnasR07} for 
computing an $O(\log{n})$-approximate vertex cover: Add vertices of degree at least $n/2$ to the vertex cover and remove them and their incident edges from the graph, namely, 
\emph{peel} these vertices; repeat this process on remaining vertices with degree threshold $n/4,n/8,\ldots,$ until all edges are covered. 

The above algorithm requires $O(\log{n})$ \emph{sequential} iterations to compute the answer.
To implement this sequential process in parallel, we use a \emph{round compression} approach as in~\cite{CzumajLMMOS17} (and implicitly in~\cite{AssadiK17}). We first partition the $O(\log{n})$ iterations of this algorithm into 
$O(\log\log{n})$ \emph{phases}: the first phase corresponds to degree thresholds between $n$ and $n^{1/2}$, the second phase to thresholds between $n^{1/2}$ and $n^{1/4}$, and so on. 
The goal now is to implement each phase in only $O(1)$ rounds of \emph{parallel} computation. To do this, we use
a \emph{random vertex partitioning} idea used by~\cite{CzumajLMMOS17}. At the beginning of each phase, we partition\footnote{Strictly speaking, our method does not necessarily leads to a ``partition'' of the vertex-set as there 
can be some small overlap between vertices across different pieces.}  the vertices of the input graph randomly into $k$ pieces and create $k$ \emph{induced subgraphs} $\Gi{1},\ldots,\Gi{k}$. Each subgraph is sent to a separate 
machine/processor which continues to run the sequential algorithm on this subgraph \emph{locally} with no further communication across the machines. After this phase ends, the machines communicate the set of peeled vertices to 
each other and update the underlying graph, i.e., remove all peeled vertices and their incident edges. Subsequent phases are implemented in a similar way. This parallel algorithm clearly can be implemented in only $O(\log\log{n})$
rounds as we only require $O(1)$ rounds per each phase. 

The high level intuition behind the round compression step is that by randomly partitioning the vertices, we can somehow ``preserve'' the degree distribution of sampled vertices even across multiple iterations of one phase. As a result, the machines should peel
the same set of vertices across. In other words, one may hope that the set of all peeled vertices by the sequential process and the union of peeled vertices across the machines in the parallel algorithm are essentially the same in each 
phase. This intuition however runs into a serious technical difficulty: the peeling process is quite sensitive to the exact degree of vertices and even slight changes in degree can move vertices between different
iterations that potentially results a cascading effect, leading to peeling very different sets of vertices across the machines. 

To address this issue, we use the techniques developed in~\cite{AssadiK17}: we design a \emph{hypothetical peeling process} which is aware of the actual minimum vertex cover of $G$ 
and show that the actual peeling process of each machine in each phase is ``sandwiched'' between two applications of this hypothetical process with different degree thresholds for peeling vertices. We then use this
to argue that the set of all vertices peeled across the machines are always contained in the solution of the hypothetical peeling process which in turn can be shown to be a relatively small set. 

\paragraph{Comparison with~\cite{AssadiK17}.} Our main idea of ``mimicing'' the sequential peeling process of~\cite{ParnasR07} for approximating minimum vertex cover in a smaller number of rounds
of parallel computation, as well as the analysis of the algorithm based on the introduction of the hypothetical process are both borrowed from~\cite{AssadiK17}. The main difference in 
our approach and~\cite{AssadiK17} lies in the idea of random vertex partitioning (which appeared first in~\cite{CzumajLMMOS17}) as opposed to the random 
edge partitioning method of~\cite{AssadiK17} (i.e., the so-called \emph{randomized composable coreset} method). It was shown in~\cite{AssadiK17} that by allocating $\Theta(n\sqrt{n})$ memory per machine
and partitioning the edges randomly across the machines, one can preserve the degree distribution of \emph{all} sufficiently high-degree vertices on each machine and implement the first phase of the actual peeling process 
in only one round of parallel computation to process all vertices of degree more than $\sqrt{n}$. Moreover, as  the remaining graph is now sufficiently sparse to be processed on a single machine of memory $\Theta(n\sqrt{n})$, 
the whole process can be implemented in a constant number of rounds. As our goal here is to use only 
$O(n)$ memory per machine, we cannot afford to preserve the degree distribution of \emph{all} vertices in every machine, neither can we simply stop after processing the first phase as the graph
is not sufficiently sparse to be stored on a single machine with $O(n)$ memory. 

\paragraph{Comparison with~\cite{CzumajLMMOS17}.} The main approach taken by Czumaj~\etal~\cite{CzumajLMMOS17} 
is also to compress the rounds of a sequential peeling process for obtaining an $O(1)$-approximation to both matching and 
vertex cover by Onak and Rubinfeld~\cite{OnakR10} (which itself is an extension of the peeling process by~\cite{ParnasR07} used in this paper)
to smaller number of MPC rounds. To achieve this, Czumaj~\etal showed that one can partition the vertices of the graph randomly across machines with $O(n)$ memory and run each phase of the peeling process in parallel
with no further communication between the machines in each phase (as is the case in this note). Similar to what argued earlier, the peeling process of~\cite{OnakR10} 
is quite sensitive to the exact degrees of vertices (even more than the algorithm of~\cite{ParnasR07}). As a result, it is a highly non-trivial challenge to argue that the parallel implementation of the algorithm 
can indeed ``faithfully'' mimic the original peeling process. To achieve this, the authors in~\cite{CzumajLMMOS17} introduce important modifications to the algorithm of~\cite{OnakR10} that allow for 
``preserving randomness'' of vertex partitioning over multiple iterations of one phase. Roughly speaking, the modified peeling process of~\cite{CzumajLMMOS17} uses a 
carefully chosen ``soft'' degree thresholding rule (rather than a fixed number as is the case in~\cite{OnakR10,ParnasR07} and our simulations in this note and previous work in~\cite{AssadiK17}) that results in 
a probability distribution for peeling vertices across the machines, whereby each vertex is peeled with almost the same probability, independent of the machine on which it resides. We refer 
the interested reader to~\cite{CzumajLMMOS17} (see Section~1.4 in particular) for more details but mention here that the algorithm and analysis in~\cite{CzumajLMMOS17} are quite intricate and require an additional number of
ideas.

\subsection{The Algorithm}\label{sec:vc-alg}

We now present our parallel algorithm \PA for approximating the minimum vertex cover. The sub-routine \LP responsible for implementing the peeling process on each machine locally is described afterwards.

\textbox{$\PA(G)$. \textnormal{A parallel algorithm for computing a vertex cover of a given graph $G$.}}{
\begin{enumerate}
	\item Define $G_1 = G$ and $\tau := O(\log\log{n})$ degree thresholds: 
	\[
		\Delta_1 := n, ~~~~ \Delta_2 := n^{1/2}, ~~ \ldots ~~ \Delta_i := n^{1/2^{i-1}}, ~~ \ldots ~~ \Delta_{\tau} := 4\log{n}.
	\] 
	\item {For} $i = 1$ to $\tau$ \emph{phases} {do} 
	\begin{enumerate}
		\item Create $k_i := \Delta_{i+1}$ graphs $\Gi{1}_i,\ldots,\Gi{k_i}_i$ as follows: 
		\begin{enumerate}[(i)]
			\item Create $k_i$ sets of vertices $\Vi{1}_i,\ldots,\Vi{k_i}_i$ whereby each set $\Vi{j}_i$ is chosen by picking each vertex in $G$ \emph{independently and uniformly at random} 
			w.p. $p_i := \frac{4\log{n}}{\Delta_{i+1}}$ (the only reason we are sampling all vertices originally in $G$ and not just the ones in $G_i$ is to simplify the math). 
			\item Each graph $\Gi{j}_i$ is the \emph{induced} subgraph of $G_i$ over vertices $\Vi{j}_i$. 
		\end{enumerate}
			
		\item\label{line:LP} {For} $j =1$ to $k_i$ \textbf{do in parallel}: $\Pi{j}_i \leftarrow \LP(\Gi{j}_i, p_i \cdot \Delta_i)$. 
		 
		\item Let $P_i \leftarrow \bigcup_j \Pi{j}_{i}$ and $G_{i+1} \leftarrow G_i \setminus P_i$. 
		\item\label{line:extra-vertices} \emph{Update} $G_{i+1}$ by removing every vertex of degree more than $\Delta_{i+1}$ (and their incident edges). Add these removed vertices to $P_i$ as well. 
	\end{enumerate}
	\item\label{line:local} Compute an $O(1)$-approximate vertex cover $P_{\tau+1}$ of $G_{\tau+1}$ on a single machine. 
	\item {Return} $P := \bigcup_{i=1}^{\tau+1} P_i$. 
\end{enumerate}
}
Several remarks are in order: first, notice that the choice of degree thresholds is such that $\Delta_{i+1} = \sqrt{\Delta_{i}}$ for each $i \in [\tau]$. Moreover, 
by Line~(\ref{line:extra-vertices}) in \PA, we always maintain the invariant that the maximum degree of the graph $G_i$ for phase $i$ is at most $\Delta_i$. 
Finally, the parallel implementation of the peeling process ends when the remaining graph is sufficiently sparse, and hence fits the memory of a single machine. At this point, we can simply 
find a vertex cover of the remaining graph using any sequential algorithm for vertex cover on a single machine. We now describe the \LP algorithm that is run by each machine locally in Line~(\ref{line:LP}) of \PA. 

\newcommand{\tmax}{\ensuremath{t_{\textnormal{\textsf{max}}}}}
\textbox{$\LP(\Gi{j}_i,\Delta)$. \textnormal{The sub-routine responsible for implementing the peeling process locally.}}{
\begin{enumerate}
		\item Define $\Gi{j}_{i,1} = \Gi{j}_{i}$ and let $\tmax$ be the smallest integer such that $\Delta/2^{\tmax} \leq 4\log{n}$. 
		\item {For} $t = 1$ to $\tmax$ {do}: 
		\begin{align*}
			\Pi{j}_{i,t} 	\leftarrow \set{\text{vertices of degree} \geq \Delta / 2^{t+1} \text{ in $\Gi{j}_{i,t}$}}, ~~~
		  	\Gi{j}_{i,t+1} \leftarrow \Gi{j}_{i,t} \setminus \Pi{j}_{i,t}. 
		\end{align*}
		\item {Return} $\Pi{j}_i := \bigcup_{t} \Pi{j}_{i,t}$. 
\end{enumerate}
}

For ease of the presentation, from now on, subscript ``$i$'' always corresponds to the phases of \PA, superscript ``$j$'' corresponds to the induced subgraphs in each phase, and subscript ``$t$'' corresponds to the iterations of \LP. 
Notice that number of phases is $\tau$, and in each phase $i \in [\tau]$, number of induced subgraphs is $k_i = \Delta_{i+1}$, and number of iterations is $\log{\paren{\Delta_{i+1}}}$. 
Our main result is that, 
\begin{theorem}\label{thm:PA}
	For any graph $G$, $\PA(G)$ outputs an $O(\log{n})$-approximation to minimum vertex cover of $G$ with probability at least $1-O(1/n)$. 
\end{theorem}

It is easy to verify that the set $P$ returned by \PA is a feasible vertex cover of $G$: any edge in $G \setminus G_{\tau+1}$ is incident on some vertex in $P \setminus P_{\tau+1}$ and $P_{\tau+1}$ is a 
vertex cover of $G_{\tau+1}$. In the next section, we prove the approximation guarantee of \PA.

\subsection{Approximation Guarantee of the Algorithm}\label{sec:vc-analysis} 

We analyze the approximation guarantee of \PA in this section. As argued earlier, the main idea behind \PA is to implement the sequential algorithm of~\cite{ParnasR07} in parallel. 
For our analysis, we introduce a \emph{hypothetical} version of this sequential algorithm with \emph{different} degree thresholds for peeling vertices in a minimum vertex cover $\Ostar$ of $G$ and the remaining
vertices $\bOstar:= V \setminus \Ostar$. Consider the following process on the original graph $G$ (defined only for analysis): 

\begin{tbox}
\begin{enumerate}
	\item Let $H_{1}$ be the bipartite graph obtained from $G$ by removing edges between vertices in $\Ostar$. 
	\item For $i = 1$ to $\tau$ phases: 
	\begin{enumerate}
		\item Let $O_i \leftarrow \emptyset$ and $\bO_i \leftarrow \emptyset$ initially. Define $H_{i,1}  := H_{i}$.  
		\item For $t = 1$ to $\ceil{\log{(\Delta_{i+1})}}$, let: 
		\begin{align*} 
		O_{i,t} &\leftarrow \set{\text{vertices in $\Ostar$ of degree} \geq {\Delta_i / 2^{t}} \text{ in $H_{i,t}$}} \\
		\bO_{i,t} &\leftarrow \set{\text{vertices in $\bOstar$ of  degree} \geq \Delta_i / 2^{t+2} \text{ in $H_{i,t}$}} \\ 
		H_{i,t+1} &\leftarrow H_{i,t} \setminus (O_{i,t} \cup \bO_{i,t}) ~~~~ O_i \leftarrow O_i \cup O_{i,t} ~~~~ \bO_i \leftarrow \bO_i \cup \bO_{i,t} 
		\end{align*}
		
		\item Let $H_{i+1} := H_i \setminus (O_{i} \cup \bO_{i})$. 
	\end{enumerate} 

\end{enumerate}
\end{tbox}

We show that \PA is ``faithfully'' mimicking this hypothetical process: with high probability, \PA does not peel more vertices from $V \setminus \Ostar$ than this hypothetical process (it may however
peel more vertices from $\Ostar$). We emphasize that this hypothetical process is \emph{only} defined for the purpose of the analysis; one cannot implement it \emph{even sequentially} without first computing
a minimum vertex cover of $G$. 

The first claim is that the set of peeled vertices by this hypothetical process itself is not much larger than a minimum vertex cover of $G$.

\begin{lemma}\label{lem:vc-hypothetical-ratio}
	$\card{\bigcup_{i=1}^{\tau} O_i \cup \bO_i} = O(\log{n}) \cdot \vc(G)$. 
\end{lemma}
\begin{proof}
	Fix any $i \in [\tau]$ and $t \in [\log{\Delta_{i+1}}]$; 
	we prove that $\bO_{i,t} \leq 8 \cdot \vc(G)$. 
	The lemma follows from this since there are at most $O(\log{n})$ different sets $\bO_{i,t}$ and the union of the sets $O_{i,t}$'s is already a subset of $\Ostar$ and hence is of size $\vc(G)$ at most. 
	
	Consider the graph $H_{i,t}$. The maximum degree in this graph is at most $\Delta_i/2^{t-1}$ by the definition of the process. Since all the edges in this graph are incident on at least one vertex of $\Vvc$, there can
	be at most $\card{\Vvc} \cdot \Delta_{i}/2^{t-1}$ edges between the remaining vertices in $\Ostar$ and $\bVvc$ in $H_{i,t}$. Moreover, any vertex in $\bO_{i,t}$ has degree at least $\Delta_i/2^{t+2}$ by definition and 
	hence there can be at most 
	\[
	\frac{\card{\Vvc} \cdot \Delta_i/2^{t-1}}{\Delta_i/2^{t+2}} \leq 8 \card{\Ostar} = 8 \cdot \vc(G)
	\]
	vertices in $\bO_{i,t}$, proving the lemma.  
\end{proof}

In the rest of this section, we prove that $\PA$~faithfully mimics this hypothetical process. Note that there is a one to one correspondence between the phases in the hypothetical process and 
phases in \PA and similarly between iterations in each phase of the hypothetical process and iterations in each run of \LP. As such, we use the term phase and iteration for both \PA and the hypothetical process. 
For any phase $i \in [\tau]$ and any graph $\Gi{j}_i$ for $j \in [k_i]$ created in \PA, we define:
\begin{tbox}
\vspace{-15pt}
\begin{align*}
	\bullet~\Oi{j}_{i,t} = O_{i,t} \cap \Vi{j}_i ~~~\bullet~\bOi{j}_{i,t} = \bO_{i,t} \cap \Vi{j}_i  ~~~\bullet~\Ai{j}_{i,t} = O_{i,t} \cap \Pi{j}_{i,t}~~~\bullet~ \Bi{j}_{i,t} = \bO_{i,t} \cap \Vi{j}_i. 
\end{align*}
\end{tbox}

Vertices in $\Oi{j}_{i,t}$ (resp. $\bOi{j}_{i,t}$) are those vertices in the graph $\Gi{j}_i$ that are peeled by the hypothetical process (over the original graph $G$) from $\Ostar$ (resp. $\bOstar$). 
On the other hand, vertices in $\Ai{j}_{i,t}$ (resp. $\Bi{j}_{i,t}$) are those vertices in the graph $\Gi{j}_i$ that are actually peeled by \LP (over the graph $\Gi{j}_i$) from $\Ostar$ (resp. $\bOstar$). We first prove a 
simple claim about the connection of the sets $\Oi{j}_{i,t}$ and $\bOi{j}_{i,t}$ and the sets $O_i$ and $\bO_i$ defined in the hypothetical process. 
\begin{claim}\label{clm:Oi-bOi}
	With probability $1-O(1/n^2)$, for all $i \in [\tau]$, 
	\[
	O_i = \bigcup_{j=1}^{k_i} \bigcup_{t=1}^{\log{(\Delta_{i+1})}} \Oi{j}_{i,t},  ~~~~ \bO_i = \bigcup_{j=1}^{k_i} \bigcup_{t=1}^{\log{(\Delta_{i+1})}} \bOi{j}_{i,t}.
	\]
\end{claim}
\begin{proof}
	The proof follows from the fact that in each phase $i \in [\tau]$, every vertex $v$ in $G$ would appear in some graph $\Gi{j}_i$ with high probability. Formally, 
	\begin{align*}
		\bigcup_{j=1}^{k_i} \bigcup_{t=1}^{\log{(\Delta_{i+1})}} \Oi{j}_{i,t} &= \bigcup_{j=1}^{k_i} \bigcup_{t=1}^{\log{(\Delta_{i+1})}} \paren{O_{i,t} \cap \Vi{j}_i} 
		= \bigcup_{t=1}^{\log{(\Delta_{i+1})}} \paren{O_{i,t} \cap \bigcup_{j=1}^{k_i} \Vi{j}_i} = O_i \cap  \bigcup_{j=1}^{k_i} \Vi{j}_i. 
	\end{align*}
	The probability that a vertex $v \in V$ is absent from $\bigcup_{j=1}^{k_i} \Vi{j}_i$ is at most, 
	\begin{align*}
	{(1-p_i)^{k_i}} \leq \exp\paren{-\frac{4\log{n}}{\Delta_{i+1}} \cdot \Delta_{i+1}} \leq 1/n^{4}. 
	\end{align*}
	By a union bound over all $n$ vertices, w.p. $1-1/n^3$, $\bigcup_{j=1}^{k_i} \Vi{j}_i = V$, proving the result for $O_i$. The equation for $\bO_i$ can be obtained exactly the same. Taking a union bound over all $\tau$
	iterations finalizes the proof. 
\end{proof}

In the remainder of this section, we condition on the event in Claim~\ref{clm:Oi-bOi}.  We further define 
\[
A_i := P_i \cap \Ostar~~~~~~\text{$and$}~~~~~~ B_i := P_i \cap \bOstar.
\]

$A_i$ (resp. $B_i$) is the set of all peeled vertices from $\Ostar$ (resp. $\bOstar$) across all parallel runs of \LP in phase $i$ of $\PA$ 
\emph{plus} the set of extra vertices added to $P_i$ in Line~(\ref{line:extra-vertices}) of \PA. 

We now establish the main connection between the sets of vertices $O_i,\bO_i$ and $A_i,B_i$. Roughly speaking, we show that union of the sets $A_i$ is a \emph{superset} of the sets $O_i$, while
union of the sets $B_i$ is a \emph{subset} of the sets $\bO_i$. 

\begin{lemma}\label{lem:vc-subset}
	For any phase $i \in [\tau]$, with probability $1-O(1/n^2)$, 
	\[
	\bigcup_{i'=1}^{i} A_{i'} \supseteq \bigcup_{i'=1}^{i} O_{i'} ~~~~~~\text{and}~~~~~~  \bigcup_{i' = 1}^{i} B_{i'} \subseteq \bigcup_{i'=1}^{i} \bO_{i'}.
	\] 
\end{lemma}
\begin{proof}
	To simplify the notation, for any $i \in [\tau]$, we define 
	\[
		O_{<i} := \bigcup_{i' < i} O_{i'} ~~~~~~\text{$and$}~~~~~~ O_{\geq i} := \Ostar \setminus O_{<i}. 
	\]
	We define these sets for $A_i$, $\bO_{i}$, and $B_{i}$ similarly. Moreover, for any $j \in [k_i]$ and $t \in [\log{\paren{\Delta_{i+1}}}]$, we define, 
	\[
		\Oi{j}_i := O_i \cap \Vi{j}_i ~~~~~~{and}~~~~~~  \Oi{j}_{i,<t} := \bigcup_{t'=1}^{t-1} \Oi{j}_{i,t'} ~~~~~~{and}~~~~~~  \Oi{j}_{i,\geq t} := \paren{\Ostar \cap \Vi{j}_i} \setminus O_{<i}. 
	\]
	Again, we define these sets similarly for $\Ai{j}_i$, $\bOi{j}_i$, and $\Bi{j}_i$. 
	
	The proof is by induction on the number of phases $i$. Define $O_0 = \bO_0 = A_0 = B_0 = \emptyset$. The base case of the induction trivially holds
	for these sets. Hence, in the following, we prove the induction step. The following lemma is the heart of the proof.  
	
	\begin{lemma}\label{lem:vc-induction}
		Fix an $i \in [\tau]$; suppose
		\[
			A_{<i} \supseteq O_{<i} ~~~~~~\text{and}~~~~~~  B_{<i} \supseteq \bO_{<i}; 
		\]
		then, with probability $1-O(1/n^2)$, for all $j \in [k_i]$ and all $t \in [\log{\paren{\Delta_{i+1}}}+1]$, 
		\[ 
			\Ai{j}_{i,<t+1} \supseteq \Oi{j}_{i,<t+1} ~~~~~~~\text{and}~~~~~~  \Bi{j}_{i,<t+1} \subseteq \bOi{j}_{i,<t+1}.
		\] 
	\end{lemma}
	\begin{proof}
	Fix an index $j \in [k_i]$. We first use the fact that the graph $\Gi{j}_{i}$ is obtained from $G_i$ by 
	sampling each vertex w.p. $p_i$ to prove that the degree distribution of sampled vertices are essentially the same in both $\Gi{j}_i$ and $G_i$ (up to the scaling factor of $p_i$). 
	In the following, we use $D(v,S)$ to denote the degree of a vertex $v$ to vertices in $S$ in the
	graph $\Gi{j}_i$. We have,

	\begin{claim}\label{clm:vc-degree-concentration}
		Fix a graph $\Gi{j}_i$ in \PA and define $\Delta := {4\Delta_{i+1} \cdot \log{n}}$. For $t \in [\log{\paren{\Delta_{i+1}}}]$: 
		\begin{itemize}
			\item For any vertex $v \in \Oi{j}_{i,t}$, $D(v,{\bOi{j}_{i,\geq t}}) \geq \Delta/2^{t+1}$ in the graph $\Gi{j}_i$ w.p. $1-O(1/n^4)$.
			\item For any vertex $v \in \bOi{j}_{i,\geq t+1}$, $D(v,{\Oi{j}_{i,\geq t}}) < \Delta/2^{t+1}$ in the graph $\Gi{j}_i$ w.p. $1-O(1/n^4)$.
		\end{itemize}
	\end{claim}
	\begin{proof}
		Fix any iteration $t \in [\log{\paren{\Delta_{i+1}}}]$ and a vertex $v \in \Oi{j}_{i,t}$. 
		By definition of $O_{i,t}$, degree of $v$ (in the hypothetical process) is at least $\Delta_i/2^{t}$ in $H_{i,t}$. Note that neighbors of $v$ in $H_{i,t}$ are precisely the vertices in
		$\bO_{i,\geq t}$ in $G$. As such, we have $D(v,\bO_{i,\geq t}) \geq \Delta_i/2^{t}$ in the graph $G$. 
		
		Next, consider the graph $G_i$. By definition, $G_i = G \setminus (A_{<i} \cup B_{<i})$. By the assumption in the lemma statement, $B_{<i} \subseteq \bO_{<i}$. As such, all vertices
		in $\bO_{i,\geq t}$ also belong to the graph $G_i$. Now, let $v$ be a vertex in $\Oi{j}_{i,t}$ and consider the neighbors of $v$ in the graph $\Gi{j}_i$, i.e., the sampled induced subgraph of $G_i$. 
		Since each vertex in $\bO_{i,\geq t}$ is sampled in $\Gi{j}_i$ w.p. $p_i$, we have, 
		\begin{align*}
			\Ex\Bracket{D(v,\bOi{j}_{i,\geq t})} = p_i \cdot {D(v,\bO_{i,\geq t})} \geq p_i \cdot \frac{\Delta_{i}}{2^{t}} = \frac{4\log{n}}{\Delta_{i+1}} \cdot \frac{\Delta_i}{2^{t}} = \frac{4\Delta_{i+1} \log{n}}{2^{t}} = \frac{\Delta}{2^{t}}.
		\end{align*}
		Moreover, as $t \leq \log{\paren{\Delta_{i+1}}}$, we know that ${\Delta}/{2^{t}} \geq 4\log{n}$. As such, by Chernoff bound, we have that, w.p. $1-O(1/n^4)$, $D(v,\bOi{j}_{i,\geq t}) \geq \Delta/2^{t+1}$ in $\Gi{j}_i$.  
		
		Similarly, for a vertex $v \in \bOi{j}_{i,\geq t+1}$, the degree of $v$ (in the hypothetical process) is smaller than $\Delta_i/2^{t+2}$ in $H_{i,t}$. This means
		that $D(v,O_{i,\geq t}) < \Delta_i/2^{t+2}$ in the original graph $G$. 
		Using the exact same argument as before, we have that w.p. $1-O(1/n^4)$, $D(v,\bOi{j}_{i,\geq t}) < \Delta/2^{t+1}$ in $\Gi{j}_i$. 
	\end{proof}
	
	Define $\Delta$ as in Claim~\ref{clm:vc-degree-concentration} and notice that this is also the threshold value used in \LP in phase $i$. 
	By using a union bound on the $n$ vertices in $G$, the statements in Claim~\ref{clm:vc-degree-concentration} hold simultaneously for all vertices of $\Gi{j}_i$ w.p. $1-O(1/n^3)$; in the following we 
	condition on this event. We are now ready to prove Lemma~\ref{lem:vc-induction}. The lemma is by induction on the number of iterations $t$. 
		
	\textbf{\emph{Base case.}} Let $v$ be a vertex that belongs to $\Oi{j}_{i,1}$; we prove that $v$ belongs to the set $\Pi{j}_{i,1}$ of \LP as well, hence $v \in \Ai{j}_{i,1}$. 
	By Claim~\ref{clm:vc-degree-concentration} (for $t=1$), the degree
	of $v$ in $\Gi{j}_i$ is at least $\Delta/4$. Note that in $\Gi{j}_i$, $v$ may also have edges to other vertices in $\Ostar$ but this can only increase the degree of $v$. 
	This implies that $v$ also belongs to $\Pi{j}_{i,1}$ by the threshold chosen in $\LP$. 
	Similarly, let $u$ be a vertex in $\bO_{i,\geq 2}$, i.e., \emph{not} in $\bO_{i,1}$; we show that $u$ is not chosen in $\Pi{j}_{i,1}$, implying that $\Bi{j}_{i,1}$ can only contain vertices in $\bOi{j}_{i,1}$. 
	By Claim~\ref{clm:vc-degree-concentration}, degree of $u$ in $\Gi{j}_i$ is less than $\Delta/4$. This implies that $u$ is not peeled in $\Pi{j}_{i,1}$. 
	In summary, we have $\Ai{j}_{i,1} \supseteq \Oi{j}_{i,1}$ and $\Bi{j}_{i,1} \subseteq \bOi{j}_{i,1}$. 
	
	\textbf{\emph{Induction step.}} Now consider some iteration $t > 1$ and let $v$ be a vertex in $\Oi{j}_{i,t}$ which does not belong to $\Ai{j}_{i,<t}$, i.e., is not peeled already. 
	By induction, we know that $\Bi{j}_{i,< t} \subseteq \bOi{j}_{i,<t}$ and hence $\Bi{j}_{i,\geq t} \supseteq \bOi{j}_{i,\geq t}$. This implies that, 
	\begin{align*}
		D(v,\Bi{j}_{i,\geq t}) \geq D(v,\bOi{j}_{i,\geq t}). 
	\end{align*}
	In other words, the degree of $v$ to $\Bi{j}_{i,\geq t}$ in $\Gi{j}_i$ is at least as large as its degree to $\bOi{j}_{i,\geq t}$. Consequently,
	by Claim~\ref{clm:vc-degree-concentration}, degree of $v$ in the graph $\Gi{j}_{i,t}$ is at least $\Delta/2^{t+1}$ and hence $v$ is peeled in $\Pi{j}_{i,t}$ (and hence belongs to $\Ai{j}_{i,t}$). 
	This implies that $\Ai{j}_{i,<t+1} \supseteq \Oi{j}_{i,<t+1}$.
	
	Similarly, fix a vertex
	$u$ in $\bOi{j}_{i,\geq t+1}$. By induction, $\Ai{j}_{i,<t} \supseteq \Oi{j}_{i,<t}$ and hence the
	degree of $u$ to $\Ai{j}_{i,\geq t}$ in $\Gi{j}_{i,t}$ is at most as large as its degree to $\Oi{j}_{i,\geq t}$; note that since $\Vvc$ is a vertex cover, $u$ does not have any other edges
	in $\Gi{j}_{i,t}$ except for the ones to $\Ai{j}_{i,\geq t}$. We can now argue as before that $u$ would not be peeled in $\Pi{j}_{i,t}$ and hence does not belong to $\Bi{j}_{i,t}$. As a result, 
	$\Bi{j}_{i,<t+1} \subseteq \bOi{j}_{i,<t+1}$.
	
	The proof of the lemma can be finalized by taking a union bound over all $k_i = \Delta_{i+1} = O(n)$ possible choices for $j \in [k_i]$. 
	\end{proof}	
	
	We are now ready to finalize the proof of Lemma~\ref{lem:vc-subset}. By Lemma~\ref{lem:vc-induction}, for $t = \tmax = \log{\paren{\Delta_{i+1}}}$, we have, 
	\[
	\bigcup_{i'=1}^{i} \bigcup_{j=1}^{k_i} \bigcup_{t=1}^{\tmax}\Ai{j}_{i',t} \supseteq \bigcup_{i'=1}^{i} O_{i'} ~~~~~~~\text{and}~~~~~~  \bigcup_{i'=1}^{i}\bigcup_{j=1}^{k_i} \bigcup_{t=1}^{\tmax}\Bi{j}_{i',t} 
	\subseteq \bigcup_{i'=1}^{i}\bO_{i'}. 
	\]
	Recall that $A_i$ is a superset of $\bigcup_{j=1}^{k_i} \bigcup_{t=1}^{\tmax}\Ai{j}_{i,t}$ and hence we already have $A_{<i+1} \supseteq O_{<i+1}$, proving this part. It thus only remains to show 
	that $B_{<i+1} \subseteq \bO_{<i+1}$ as well. To do so, we need to argue that the new set of vertices from $\bOstar$ 
	added in Line~(\ref{line:extra-vertices}) all belong to $\bO_{<i+1}$. 
	
	Consider any vertex $v \in \bO_{\geq i+1}$. 
	We know that degree of $v$ to $O_{\geq i}$ is at most $\Delta_{i}/2^{\tmax+2} < \Delta_{i+1}$ as otherwise $v$ would be peeled in the last iteration. By the previous argument as $A_{<i+1} \supseteq O_{<i+1}$, 
	this implies that	the degree of this vertex in Line~(\ref{line:extra-vertices}) is smaller than $\Delta_{i+1}$ as well. Hence, even after adding the vertices with degree at least $\Delta_{i+1}$ to $P_i$ 
	in Line~(\ref{line:extra-vertices}), $P_i \cap \bOstar \subseteq \bO_{\leq i}$, finalizing the proof. 
\end{proof}

We are now ready to prove Theorem~\ref{thm:PA}. 

\begin{proof}[Proof of Theorem~\ref{thm:PA}]
	By Lemma~\ref{lem:vc-subset} for $i = \tau$, we have $\bigcup_{i = 1}^{\tau} B_{i} \subseteq \bigcup_{i=1}^{\tau} \bO_{i}$. Additionally, by Lemma~\ref{lem:vc-hypothetical-ratio}, we have 
	$\card{\bigcup_{i=1}^{\tau} \bO_{i}} = O(\log{n}) \cdot \vc(G)$. As a result $P \cap \bOstar$ is of size at most $O(\log{n}) \cdot \vc(G)$. The final result now follows since $P \cap \Ostar$ can be of size  
	at most $\card{\Ostar} = \vc(G)$. 
\end{proof}


\subsection{MPC Implementation of \PA}\label{sec:mpc-PA}

We show here that \PA can be implemented in the MPC model with $\Ot(n)$ space per each machine. 
The main part of the argument is to show that the space on each machine is enough to run the \LP in Line~(\ref{line:LP}) of \PA. To do this, it suffices to show that, 

\begin{lemma}\label{lem:Gij-size}
	With probability $1-O(1/n^2)$, for any $i \in [\tau]$ and any $j \in [k_i]$, the number of edges in $\Gi{j}_i$ is $O(n \log^{2}{n})$. 	
\end{lemma}
\begin{proof}
	By Line~(\ref{line:extra-vertices}) of \PA, we have the invariant that at the beginning of each phase $i \in [\tau]$, the maximum degree of the graph $G_i$ is at most $\Delta_i$. Hence, 
	the expected maximum degree of the graph $\Gi{j}_i$ is at most $p_i \cdot \Delta_i = \frac{4\log{n}}{\Delta_{i+1}} \cdot \Delta_i = 4\Delta_{i+1} \cdot \log{n}$. By Chernoff bound, with probability $1-O(1/n^4)$, 
	the maximum degree of $\Gi{j}_i$ is $O(\Delta_{i+1} \cdot \log{n})$. Using another application of Chernoff bound, we also have that the number of vertices assigned to $\Gi{j}_i$ is at
	most $O(p_i \cdot n) = O(n\log{n}/\Delta_{i+1})$, with probability $1-O(1/n^4)$. As a result, the total number of edges in $\Gi{j}_i$ is $O(n\log^{2}{n})$ with this probability. Taking a union bound over all
	possible $k_i = O(n)$ indices $j$ and $O(\log\log{n})$ choices for $i$ finalizes the proof. 
\end{proof}

It is also easy to see that the graph $G_{\tau+1}$ fits the memory of a single machine and hence last step of \PA can be implemented locally. 
The rest of \PA can be implemented in the MPC model using standard techniques similar to~\cite{CzumajLMMOS17}; we refer the reader to Section 5 of~\cite{CzumajLMMOS17} for details of this implementation. 

\subsection{Extension to Smaller Memory Requirements} \label{sec:mpc-smaller-memory}

We now briefly describe the necessary changes required to make \PA work when the memory of each machine is some fixed parameter $s = f(n) = n^{\Omega(1)}$. Here,  $f(n)$ can be even sublinear in $n$, i.e., $f(n) = o(n)$.
For simplicity of exposition, we assume that the memory on each machine is  $O(s) \cdot \poly\log{(n)}$ (as opposed to exactly $s$), a simple rescaling of the parameters shows the result when the memory is exactly $s$. 
There are only two changes that need to be done in \PA: 
\begin{enumerate}
	\item Define the degree thresholds in the first line of \PA as $\tau := O(\log\log{n})$ thresholds: 
	\[
		\Delta_1 := n, ~~~~ \Delta_2 := \frac{n}{s^{1/2}}, ~~ \ldots ~~ \Delta_i := \frac{n}{s^{1-1/2^{i-1}}}, ~~ \ldots ~~ \Delta_{\tau} := 4\paren{\frac{n}{s}}\cdot\log{n}.
	\] 
	\item In Line~(\ref{line:local}) of \PA, instead of computing an approximate minimum vertex cover of $G$ on a single machine, directly simulate the original peeling process using $O(1)$ MPC rounds per each iteration (not phase). 
\end{enumerate}

The proof of correctness of this algorithm is exactly the same as the one for the original algorithm.  Moreover, the first part of the algorithm, i.e., implementing the peeling process in parallel still requires $O(\log\log{n})$ MPC rounds
and (by the same argument in Lemma~\ref{lem:Gij-size}) $\Ot(s)$ memory per machine. 
Finally, implementing the last step, i.e., the modified version of Line~(\ref{line:local}) of \PA, requires $\Ot(s)$ memory per machine and $O(\log{(\frac{n}{s}}))$ MPC rounds. This is because the 
maximum degree of $G_{\tau+1}$ is at most $\Ot(n/s)$ and hence $O(\log{(\frac{n}{s}}))$ iterations of the peeling process suffice to solve the problem and each iteration can
 be implemented in $O(1)$ MPC rounds using standard techniques. We refer the reader to~\cite{CzumajLMMOS17} (see Lemma~5.1)
for more details on the implementation. To conclude, we obtain that, 
\begin{theorem}\label{thm:vc-full}
	There exists an MPC algorithm that with high probability computes an $O(\log{n})$ approximation to minimum vertex cover in $O(\log\log{n} + \log{\paren{\frac{n}{s}}})$ rounds, assuming that the memory per each machine is 
	$s = n^{\Omega(1)}$.  
\end{theorem}